\documentclass[lettersize,journal]{IEEEtran}
\usepackage{amsmath,amssymb,amsfonts}
\usepackage{array}
\usepackage{cite}
\usepackage{amsmath,amssymb,amsfonts}
\usepackage{algorithmic}
\usepackage{graphicx}
\usepackage{textcomp}
\usepackage{mathtools}				
\usepackage{bm}  					
\usepackage[ruled,linesnumbered]{algorithm2e}
\usepackage{xcolor}
\usepackage{color}
\usepackage{stfloats}
\usepackage{endnotes}
\usepackage{url}

\usepackage{amsthm}
\newtheorem{theorem}{Theorem}
\newtheorem{lemma}{Lemma}

\newtheorem{definition}{Definition}  
\newtheorem{assumption}{Assumption}

\allowdisplaybreaks[4]

\newcommand{\T}{\top}					
\usepackage{xcolor}

\begin{document}
	
	\title{SSGCNet: A Sparse Spectra Graph Convolutional Network for Epileptic EEG Signal Classification}
	
	\author{Jialin Wang,
		Rui Gao,~\IEEEmembership{Member,~IEEE,}
		Haotian Zheng, Hao Zhu, 
		and~C.-J. Richard Shi,~\IEEEmembership{Fellow,~IEEE}
		\thanks{J. Wang, R. Gao, H. Zheng and H. Zhu are with the State Key Laboratory of ASIC and Systems, the Institute of Brain-Inspired Circuits and Systems, and Zhangjiang Fudan International Innovation Center, Fudan University, Shanghai, China, 201203. R.~Gao is also with the Department of Electrical Engineering and Automation, Aalto University, Espoo, Finland, 02150. C.-J. Richard Shi is with the Department of Electrical and Computer Engineering, University of Washington, Seattle, USA, WA 98195 (E-mail: jialinwang16@fudan.edu.cn; rui.gao@aalto.fi; 20212020178@fudan.edu.cn; 20112020135@fudan.edu.cn; cjshi.fudan@gmail.com). }
	}

	\markboth{Journal of \LaTeX\ Class Files,~Vol.~14, No.~8, August~2021}
	{Shell \MakeLowercase{\textit{et al.}}: A Sample Article Using IEEEtran.cls for IEEE Journals}

	\maketitle
	
	\begin{abstract}
		In this article, we propose a sparse spectra graph convolutional network (SSGCNet) for solving Epileptic EEG signal classification problems. The aim is to achieve a lightweight deep learning model without losing model classification accuracy. We propose a weighted neighborhood field graph (WNFG) to represent EEG signals, which reduces the redundant edges between graph nodes. WNFG has lower time complexity and memory usage than the conventional solutions. Using the graph representation, the sequential graph convolutional network is based on a combination of sparse weight pruning technique and the alternating direction method of multipliers (ADMM). Our approach can reduce computation complexity without effect on classification accuracy. We also present convergence results for the proposed approach. The performance of the approach is illustrated in public and clinical-real datasets. Compared with the existing literature, our WNFG of EEG signals achieves up to $10$ times of redundant edge reduction, and our approach achieves up to $97$ times of model pruning without loss of classification accuracy.
	\end{abstract}
	
	\begin{IEEEkeywords}
		Graph neural network, electroencephalogram (EEG) signal classification, weight pruning, alternating direction method of multipliers (ADMM),  nonconvextiy.
	\end{IEEEkeywords}

	\section{Introduction}
	\label{sec:introduction}
	\IEEEPARstart{E}{pilepsy} is one of the most common neurological diseases, which is accompanied by super-synchronous abnormal discharge of electroencephalogram (EEG) signals~\cite{zhu2014epileptic,mohammadpoory2017epileptic}. Generally, doctors need to detect abnormal EEG signals from epileptic seizures from dozens of hours of EEG signals. However, the abnormal EEG signal duration of epileptic seizures is pretty short (only a few seconds)~\cite{hassan2016epileptic}. As a result, it becomes significant to detect abnormal signals from massive EEG data automatically. In this article, the objective is to develop an autonomous effective method to classify epileptic EEG signals.

	The effective representation of EEG signals is an essential prerequisite for distinguishing different categories of EEG signals~\cite{wang2017eeg238,zou2018complex}. Due to the complexity of EEG waveforms, most methods usually have the limitation on representing the relationship between EEG signal sampling points, which makes these methods challenging to extract relationship features between sampling points~\cite{supriya2018eeg,yang2019complex,iacovacci2019visibility}. In recent years, converting signals into graphs has received extensive attention~\cite{zou2018complex, wang2017eeg238}. Although the graph representation based methods can effectively extract the relationship between signal sampling points, these methods have limitations in building effective connections between sample points, especially when there exist a lot of redundant edges in the graph representation. That leads to an enormous generation time and memory usage, that may limit the promotion of graph representation methods on lightweight hardware devices.

	More recently, traditional machine learning methods have been used for EEG signal classification tasks, for example, empirical mode decomposition method, hybrid type machine learning method, and logistic model tree based method~\cite{bajaj2011classification,subasi2019epileptic,kabir2016epileptic}.
	However, most methods require manual feature selection on EEG signals, which heavily relies on researcher's experience and domain knowledge of EEG signals.
	These methods become ineffective once the experience is overly subjective~\cite{cai2018epileptic}. Therefore, it is of great value to propose an autonomous classification method for EEG epilepsy with feature selection.

	Deep learning methods have received extensive attention in autonomous classification~\cite{zhou2018epileptic,hussein2018epileptic}. End-to-end deep learning models can independently extract useful features from data~\cite{asif2020seizurenet,wei2018automatic}. The features captured by the models do not need the scope of knowledge mastered by researchers, and then the models can construct the mapping relationship between input and output signals automatically. Although the models can extract the features in the data independently, it is difficult to directly obtain the hidden features from the original data~\cite{acharya2018deep,Li2018ECCV}. For extracting a large number of features from datasets, large-scale deep learning models have been developed in these years~\cite{Xu2019How,Errica2020A}. With increasing model scale, they have to take a huge of computing resources and memory storage, which leads that EEG signal classification scenarios can not be deployed on practical low-power device. Hence, it becomes promising to optimize lightweight deep learning models arising in epilepsy EEG signal classification.

	In this article, we propose a sparse spectra graph convolutional network (SSGCNet) for Epileptic EEG signal classification. We first represent EEG signals as a frequency domain graph representation, and then use the sequential convolutional module to extract features between graph nodes. Under sparsity constraints, we introduce the ADMM-type splitting and weight pruning strategy, which can compress the model without reducing the classification accuracy. We also analysis the convergence results of the proposed methods. Experimental results demonstrate the promising performance of our methods in various real-world datasets. The main contributions of the article are summarized as follows: 
	
	\begin{enumerate}
	
		\item We present a weighted neighborhood field graph (WNFG) representation method to represent epilepsy EEG signals.
		
		\item We develop a sparse spectra graph convolutional neural network (SSGCNet) model, which effectively extracts the node relationship features and sequential features.

		\item We propose an efficient ADMM weight pruning method, which significantly achieves a high compression rate without loss of classification accuracy.
		
		\item We analyze the convergence results of the proposed method, and then apply the proposed model to several real-world application.

	\end{enumerate}
	
	The main advantage of our model is that the computational cost and space occupancy rate is much less than in other traditional methods. The average redundant edge of our graph representation is reduced by $10$ times on public datasets and $10$ times on clinical-real datasets, respectively. Our model compression exceeds other deep learning models in the Epilepsy EEG signal classification around $97$ times.

	The rest of this article is structured as follows. 
	We conclude this section by reviewing optimization methods on deep learning and the alternating direction method of multipliers (ADMM). In Section~\ref{sec:data_rep}, we introduce the graph representation of EEG signals, where the model can effectively extract the node relationship features in the sparse graph representation. The proposed sparse spectra graph convolutional neural network (SSGCNet) method is presented in Sections~\ref{sec:method}. In Section~\ref{sec:results}, by public and clinical-real datasets, various experimental results demonstrate the effectiveness and accuracy. Section~\ref{sec:Conclusion} draws the concluding remarks.

	\subsection{Optimization on Deep Learning}
	The optimization of the deep learning EEG signal classification method can usually be divided into data representation based optimization methods and weight pruning methods.
	
	1) For data representation methods, this kind of methods is to represent EEG signals in the time domain or frequency domain~\cite{liu2013model,cai2018epileptic}. The methods can be used to represent different characteristics of EEG signals. However, the relationship between the sampling points of the EEG signal is always ignored~\cite{wang2017eeg238,zou2018complex,zhu2014epileptic}. The graph representation based method has recently attracted much attention~\cite{supriya2018eeg,yang2019complex}. These methods use each data point of the EEG signal as a node of the graph structure, and convert the EEG signal into the graph structure through connection rules~\cite{supriya2021epilepsy}. The graph representation of the signal can effectively represent the characteristics of the relationship between the different sampling points. Since these graph structures have many redundant edges, such graph representation occupies large computation space to store redundancy of weights. For this reason, we design a weighted neighborhood field graph (WNFG) based on the EEG signal graph representation, which can significantly reduce the redundant edges between sampling points.

	2) For model weight pruning methods, the existing methods usually take the model as a whole module~\cite{ye2018progressive,srivastava2014dropout,baldi2013understanding} and optimize redundant weights with small values in the parameter distribution of the model~\cite{zhang2018systematic,li2016pruning,yim2017gift,wang2020seizure}. Some methods are optimized on multilayer perceptrons and convolutional neural networks, and ignore the sparsity property of arising in the deep learning model~\cite{liu2018rethinking,blalock2020state,see2016compression}. Some methods introduce sparse weight pruning into the model, but they methods lack convergence analysis \cite{ye2018progressive,see2016compression}. With the development of optimization methods, variable splitting methods has recently been introduced in the arena of sparsity type problems, where it has become increasingly important due to splitting property~\cite{Splitting2017book,Gao2019ieks,Gao2020autonomous}. The methods such as ADMM~\cite{Boyd2011admm,Gao2017analysis}, augmented Lagrangian splitting (ALS)~\cite{Goldstein2009Split}, and Peaceman-Rachford splitting (PRS)~\cite{He2014peaceman} are efficient methods that can tackle this kind of sparsity problems. They effectively decompose the original problem into a sequence of simple subproblems. We present a new graph convolutional network with based on the ADMM-type splitting and weight pruning strategy.

	\subsection{Alternating Direction Method of Multipliers (ADMM)}
	\label{sec:vs}
	
	This subsection briefly reviews ADMM. Consider a minimization optimization problem
	\begin{equation}\label{eq:vs_function}
		\begin{split}
			\begin{aligned}
				\min_{\mathbf{x}} \, f(\mathbf{x}),  \quad
				\mathrm{s.t.}\, \mathbf{x} \in \mathcal{C},
			\end{aligned}
		\end{split}
	\end{equation}
	where $f({\mathbf{x}})$ is the cost function, and $\mathcal{C}$ is a constraint set. The set $\mathcal{C}$ can be boolean constraints, the set of matrices with rank, or the number of non-zero elements \cite{Splitting2017book}.

	One approach to solving problems of the form \eqref{eq:vs_function} is based on ADMM~\cite{Boyd2011admm,Gao2019ieks}. ADMM is a broadly used variable splitting method for solving constrained optimization problems \cite{Gao2020Variable}. Instead of minimizing the given constrained minimization problem, ADMM can transform it into a unconstrained one. The constraints are then replaced by adding an indicator function. Its main advantage is the splitting ability that splits a complex problem into several simpler subproblems. Mathematically, we first write \eqref{eq:vs_function} as an equality constrained optimization problem
	\begin{equation}\label{eq:vs_function2}
		\begin{split}
			\min_{\mathbf{x}}\, &f(\mathbf{x}) + {g}(\mathbf{z}) \\
			\mathrm{s.t.}\  &\mathbf{x} = \mathbf{z}, 
			\quad g(\mathbf{z})= \begin{cases}
				\mathbf{x}, &   \mathbf{z} \in \mathcal{C}, \\
				\infty, &     \text{otherwise}.
			\end{cases}
		\end{split}
	\end{equation}
	Here, ${g}(\mathbf{z})$ is the indicator function. We then define an augmented Lagrangian function $\mathcal{L}(\mathbf{x},\mathbf{z};\bm{\eta})$ of the minimization problem \eqref{eq:vs_function2} 
	\begin{equation}\label{eq:Lagrangian_function}
		\begin{split}
			\begin{aligned}
				\mathcal{L}(\mathbf{x},\mathbf{z};\bm{\eta}) =   
				f (\mathbf{x}) + g ( \mathbf{z})  
				+ \bm{\eta}^\T(\mathbf{x}-\mathbf{z})
				+ \frac{\rho}{2} \| \mathbf{x}- \mathbf{z}\|^2,
			\end{aligned}
		\end{split}
	\end{equation}
	where $\bm{\eta}$ is a Lagrange multiplier and $\rho$ be a penalty parameter. Given the initial iteration $\mathbf{x}^{(0)}$, $\mathbf{z}^{(0)}$, $\bm{\eta}^{(0)}$, the ADMM iteration has the update steps
	\begin{subequations}
		\label{eq:general-admm}
		\begin{align}
			\label{eq:x-primal-admm}
			\mathbf{x}^{(k+1)} &= \!\mathop{\arg\min}_{\mathbf{x}} f(\mathbf{x}) \!+\! (\bm{\eta}^{(k)})^\T(\mathbf{x} -  \mathbf{z}^{(k)})
			\!+\! \frac{\rho}{2} \left\|\mathbf{x} \!- \! \mathbf{z}^{(k)}\right\|^2, \\
			\label{eq:z-admm}
			\mathbf{z}^{(k+1)} &=\mathop{\arg\min}_{\mathbf{z}} g(\mathbf{z}) 
			+({\bm{\eta}^{(k)}})^\T (\mathbf{x}^{(k+1)} -  \mathbf{z} ) \nonumber \\
			&\quad  \quad \quad + \frac{\rho}{2} \left\|\mathbf{x}^{(k+1)} -  \mathbf{z} \right\|^2, \\
			\label{eq:eta-admm}
			\bm{\eta}^{(k+1)} & = \bm{\eta}^{(k)} +   \rho \left(\mathbf{x}^{(k+1)} - \mathbf{z}^{(k+1)} \right),
		\end{align}
	\end{subequations}
	where $k$ is the iteration number. The minimization problem can be solved by alternating the updates of all the variables. In this article, we apply the ADMM method to the sparsity version of the proposed model, associating with the steps of model weight pruning.

	\begin{figure*}[htbp]
		\vspace{-0cm}
		\setlength{\abovecaptionskip}{-0cm}
		\setlength{\belowcaptionskip}{-0cm}
		\centering \includegraphics[width=\textwidth]{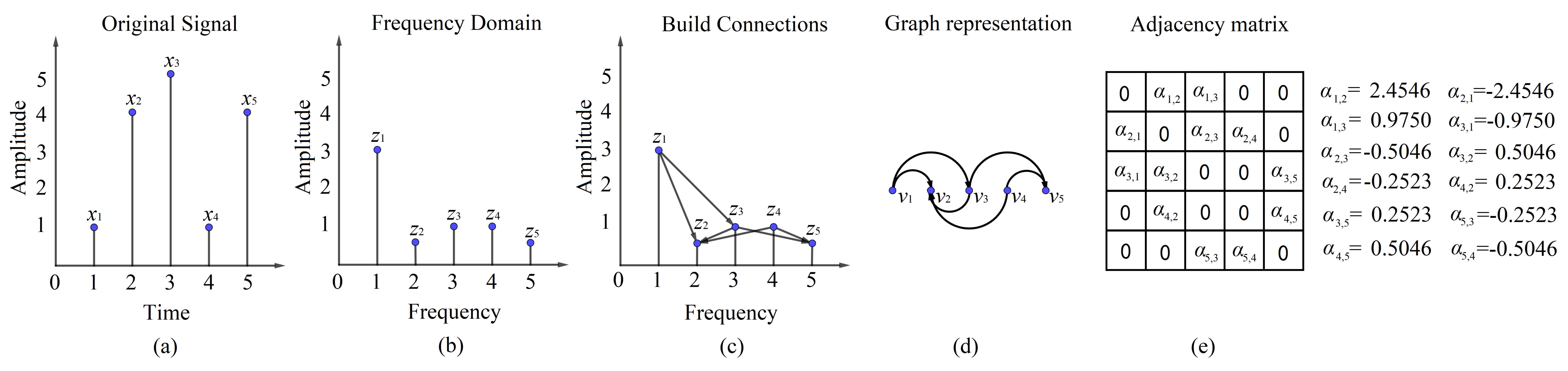}
		\caption{(a) Original signal; (b) the frequency domain signal; (c) the process of building connections; (d) the graph representation of the signal; (e) the adjacency matrix of the graph.
		}
		\label{fig:wog_wnfg}
	\end{figure*}

	\section{Graph Representation}
	\label{sec:data_rep}
	
	In this section, we introduce a new graph representation method arising in the proposed sparse spectra graph convolution network (SSGCNet). Traditional graph representation methods, for example the weighted overlook graph (WOG) method~\cite{wang2020weighted}, have a time complexity of $O(n^2)$. The main reason is that there are huge numbers of redundant edges of the WOG graph structure. We design a lightweight structure to represent the EEG signals.

	Let the vector $\mathbf{x}=\{x_1,x_2,x_3,\ldots,x_n\}$ contain $n$ sampling points. 
	Inspired by our previous work~\cite{wang2020sequential}, we adopt the preprocessing operation of transforming the time series into the frequency domain, given by
	\begin{equation} \label{eq:x_i}
		z_m = \sum_{i=1}^{n}{x_ie^{-\bar{j}\frac{2\pi}{n}im}}, \, m=1,2,\ldots,n.
	\end{equation}
	Here, $z_m$ is the frequency domain of the $i$-th element of the signal $\mathbf{x}$ (see Fig.~\ref{fig:wog_wnfg} (b)), and $\bar{j}$ is the value from imaginary part.
	Such the transform can effectively reduce the intra-class differences between EEG signal samples. 
	
	Then, we represent the frequency domain signal into a graph structure.
	Mathematically, the graph can be defined as $\mathcal{G}=\{\mathcal{V}, \mathcal{E}\}$, where the set of nodes $\mathcal{V}= \{v_1,v_2,v_3,\ldots,v_n\}$ corresponds to the elements of the signal $\mathbf{x}$. The set of edges $\mathcal{E}=\{E_1,E_2,E_3,\ldots,E_n\}$ is related to each node, and each edge $E_i$ of the $i$-th node has the weight set $\{\alpha_{i,1},\alpha_{i,2},\alpha_{i,3},\ldots,\alpha_{i,n}\}$.
	The adjacency matrix $\mathbf{A}$ can be then written as
	\begin{equation}\label{eq:adjacency}
		\mathbf{A}_{i,j}= \alpha_{i,j}. \\
	\end{equation}
	Particularly, if the weighted graph is constructed and has an edge between node $i$ and node $j$, then $\alpha_{i,j}$ is the value of the connection relationship between double notes, otherwise the weight $\alpha_{i,j} = 0$.
	In particular, when it belongs to a weightless graph and exists an edge between nodes, $\alpha_{i,j} = 1$; otherwise, $\alpha_{i,j} = 0$.

	As shown in Fig.~\ref{fig:wog_wnfg} (c), we set the connections of nodes within the specified distance $K$ with $K \ll n$. The time complexity of the graph structure then becomes $O(Kn)$. We now obtain a graph structure with lower complexity. In particular, when $K = 1$, the graph structure of each signal segment is a chain structure. When $K$ is unreasonably small, the deep learning model is difficult to converge, which will cause the loss of the model to oscillate. In our article, the value of $K$ is set in the range $[20,\,25]$, which simplifies the balanced graph structure and improves the stability of the classification accuracy.

	For any two data points $z_i, \,z_j\in \{z_m\}_{m=1}^n$, the distance range between $i$ and $j$ is less than $K$. As shown in Fig.~\ref{fig:wog_wnfg} (d), the connection rule between different points is defined as
	\begin{equation} \label{eq:connection}
		z_i>z_j \,\text{with} \, (|i-j|<K).
	\end{equation}
	The value of the edge between the data points is $\alpha_{i,j}$, which has the following equation 
	\begin{equation}\label{eq:alpha}
		\alpha_{i,j} = \frac{z_i - z_j}{i - j}.
	\end{equation}
	Here, if and only if $z_j$ is not equal to $z_i$, otherwise $\alpha_{i,j}=-\alpha_{j,i}$. After computing all $\alpha_{i,j}$, we obtain the whole adjacency matrix $\mathbf{A}$ of the graph (see Fig.~\ref{fig:wog_wnfg} (e)). 
	
	As mentioned above, the entire adjacent matrix of WOG is full of connected edges, which takes high time complexity. We use positive and negative signs as the connection direction for the direct selection. As a result, our graph structure has an effective edge connection and low computation complexity. This is also consistent with the relationship between EEG signal rhythms, which is adaptable to real-world applications. In the following, we will propose the graph classification method for EEG signals.

	\begin{figure*}[htbp]
		\setlength{\abovecaptionskip}{-0cm}
		\setlength{\belowcaptionskip}{-0cm}
		\centering \includegraphics[width=\textwidth]{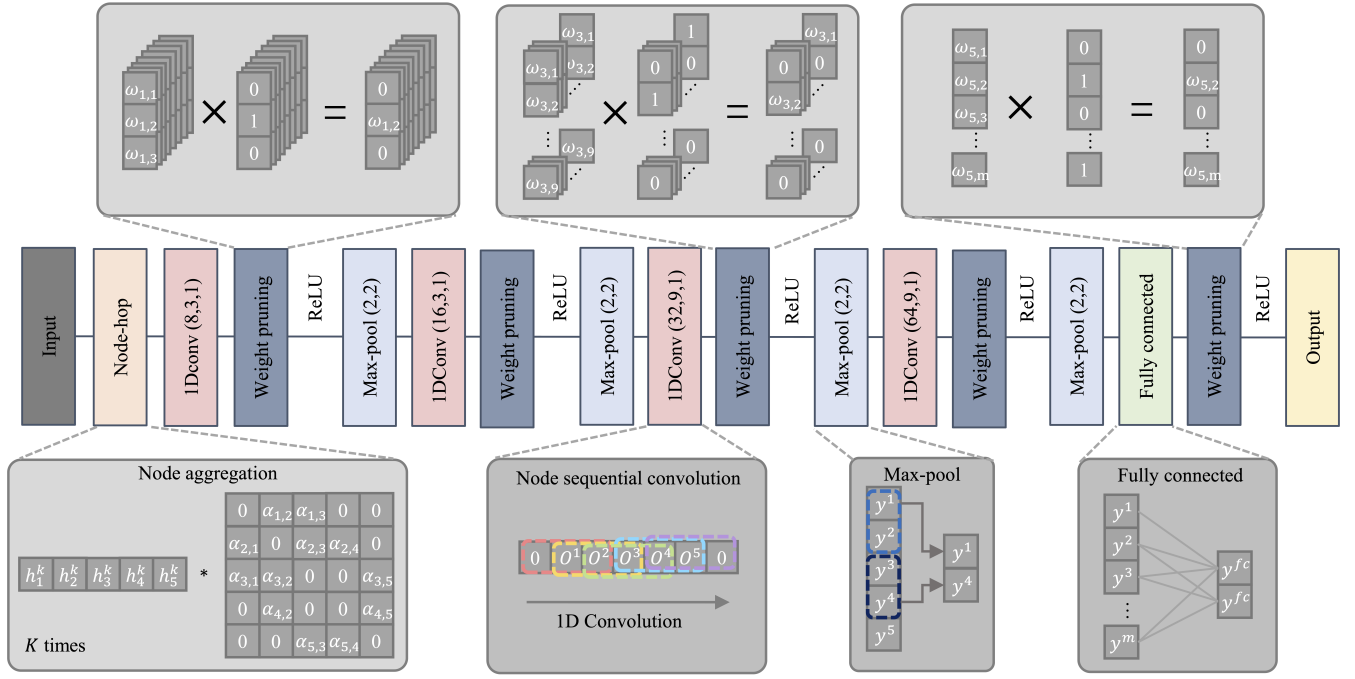}
		\caption{The framework of SSGCNet consists of the following main parts: node aggregation, node sequential convolution, fully connected layer and weight pruning. Note that we only perform weight pruning operations on the node sequential convolution layer and the fully connected layer.}
		\label{fig:network}
	\end{figure*}

	\section{Graph Classification Method}
	\label{sec:method}
	
	In this section, we propose the sparse spectra graph convolutional neural network (SSGCNet) for Epileptic EEG signal classification.

	\subsection{Problem Formulation}
	
	Let $\mathbf{w}_{l} \in \mathbb{R}^{n_w}$ be a weight parameter of the $l$-th layer and $\mathbf{y} \in \mathbb{R}^n$ be the input data. The problem of training an $N$-layer deep learning model can be formulated as:
	\begin{equation}\label{eq:optimization}
		\begin{split}
			\begin{aligned}
				\{\hat{\mathbf{w}},\hat{\mathbf{y}}\}&=\mathop{\arg\min}_{\mathbf{w}_{1:N}}
				f (\mathbf{w}_{1:N}, \mathbf{y}) \\
				{\mathrm{s.t.}}\, 
				&   
				\operatorname{\textbf{card}}(\mathbf{\Omega} \mathbf{w}_l ) \leq \text{const}, \,
				l = 1,\ldots,N, 
			\end{aligned}
		\end{split}
	\end{equation}
	where $\hat{\mathbf{w}}\in \mathbb{R}^{n_w}$ is the optimal weight parameter sequence, $\hat{\mathbf{y}}\in \mathbb{R}^{n_y}$ is the output in deep learning model corresponding to the probability of seizure or non-seizure, $\mathbf{\Omega} \in \mathbb{R}^{n_o \times n_w}$ is an operator, $f$ is the loss function, $\mathbf{w}_{1:N}$ represents the sequence of $\{\mathbf{w}_1,\ldots,\mathbf{w}_N\}$ and $\text{const}$ represents a constant. Our goal is to obtain the output $\hat{\mathbf{y}}$ and the parameters $\hat{\mathbf{w}}$. 
	
	The formulation~\eqref{eq:optimization} permits the usage of kinds of loss functions such as root mean square error, structure similarity index measure, min-max function or cross-entropy type loss function~\cite{krizhevsky2012imagenet}. The settings of the operator $\mathbf{\Omega}$ can be used to represent different pruning methods. For example, when $\mathbf{\Omega}$ is an identity matrix ($n_o = n_w$), the constraint set $\{\operatorname{\textbf{card}}(\mathbf{\Omega} \mathbf{w}_l ) \leq \text{const}\}$ represents the number of non-zero elements of the parameter $\mathbf{w}_l$. When the number is less enough, that means the weight parameters are sparse. When the matrix $\mathbf{\Omega}$ is structured, the weight can be pruned in parallel. It should be noted that, if all the products $\mathbf{\Omega} \mathbf{w}_l$ is out of the constraint sets, the objective becomes training a deep learning model without any weight pruning.
	
	Due to the nonconvexity, it is challenging to solve such the problem, {particularly when the parameters and the signals are in large-scale size}. To address the issue, we introduce the SSGCNet which combines the ADMM-type splitting and weight pruning strategy.
	
	\subsection{The Framework of SSGCNet}
	Here we introduce the framework of SSGCNet. The SSGCNet includes two-hop aggregation operations, four sequential convolutional layers, two fully connected layers, and weight pruning process. The aggregation operation is performed by multiplying the vector by the adjacency matrix. The sequential convolutional layers can accurately extract sequential features in the sequence. The fully connected structure can be equivalent to the readout structure in the graph neural network. The whole framework of our SSGCNet is shown in Fig.~\ref{fig:network}. Although the node aggregation of the graph neural network belongs to a kind of multi-hop operation for all nodes, our graph structure can completely retain the information of graph node aggregation in this process. We construct the aggregated part of the graph nodes as a single module for operations on the convolutional layer and fully connected layer. Then, we use the sparse weight pruning method to optimize the structure of the model, which will be discussed in Section \ref{sec:admm-dl}.

	Mathematically, the graph $\mathcal{G}$ with $n$ nodes is equivalent to the input signal containing $n$ data points. We first obtain the primary aggregation vector of each node and its neighbour nodes through the product of the vector $\mathbf{h}=\{h^{k}_{v_1},h^{k}_{v_2},h^{k}_{v_3},\ldots,h^{k}_{v_n}\}$, where $\mathbf{h}^{k}$ denotes the value of $\mathbf{h}$ at $k$-th iteration.  In particular, when $k = 0$, the vector $\mathbf{h}$ is an all-ones vector. With performing the $k$-hops, we use the vector $\mathbf{h}$ and the adjacency matrix $\mathbf{A}$ to multiply $k$ times. At the $k$-th time, the relative position of each node in the vector remains unchanged. We then have 
	\begin{equation} 
		h^{k}_{v_i} = {h^{k-1}_{v_i} + \sum\limits_{\{u_i\}} h^{k-1}_{u_i}}, \\
	\end{equation}
	where $v_i$ denotes the $i$-th node of the graph $\mathcal{G}$, and $u_i$ denotes the set of all the neighborhood nodes of $v_i$. After obtaining the node aggregation vector, $h^{k}_{v_i}$ is multiplied by each element of the learnable parameter $\bm{\theta} = \{\theta^1,\theta^2,\ldots,\theta^n\}$, given by
	\begin{equation}\label{eq:aggregation}
		o^i = \theta^i \cdot h^{k}_{v_i},
	\end{equation}
	which effectively weighs the importance of different node information in backpropagation. Note that the operation does not change the node order of the aggregation vector.
	
	The aggregation vector is convolved through the one-dimensional convolution kernel. The size of the convolution kernel is the receptive field of the node range. One-dimensional convolution uses the operation of zero padding on both sides. Note that the size of the output vector is consistent with the size of the input vector after the convolution is completed by  
	\begin{equation} \label{eq:convolution}
		y^i = \sigma\left({\sum\limits_{i=1}^{n} {w}^{u}_l\cdot o^{i+u+1}+b_l}\right),
	\end{equation}
	where $\sigma(\cdot)$ is the ReLU activation function~\cite{GlorotBB11}, $u$ is the kernel size, ${w}^{u}_l$ is the convolution kernel parameter in $l$-th layer, and $b_l$ is the bias in the $l$-th layer.
	
	After that, we perform a max-pooling operation on the two results of the convolutional layer, expressed as
	\begin{equation}
		y^{i} = \max \left(y^{2i-1}, \, y^{2i}\right),
	\end{equation}
	and then we input the signal to the fully connected layer by
	\begin{equation} \label{eq:y_fc}
		{y}_i^{fc} =  \sigma\left( \sum\limits_{i=1}^{n} {w}_l\cdot y^i + {b}_l\right).
	\end{equation}
	Here, ${y}^{fc}$ is the output of the SSGCNet. Note that we use Adam optimizer to optimize the model~\cite{kingma2014adam}, and the function $f$ is a cross-entropy function, defined by
	\begin{equation} \label{eq:cross-entropy}
		\begin{aligned}
			f(\mathbf{w}_{1:N},\mathbf{y}) = - \log \left(\frac{\exp^{y^{fc}_i}}{ \sum_{i=1}^c \exp^{y^{fc}_i}}\right).
		\end{aligned}
	\end{equation}
	
	Until now, we can obtain the output $\hat{\mathbf{y}}$ by computing~\eqref{eq:cross-entropy}. Our model is suitable for signals that can be well-represented in graph nodes. However, {as the scale of graph nodes and the computational complexity of the networks increase, the methods lack strong performance guarantee. In the following, we leverage the sparse redundancy in the number of weights of SSGCNet.}
	
	\subsection{ADMM Weight Pruning Method}
	\label{sec:admm-dl}

	Using the ADMM-type splitting method~\cite{Boyd2011admm,Gao2017analysis}, we introduce an auxiliary variable sequence $\{\mathbf{z}\}_{l=1}^N$ and an indicator function $g(\cdot)$. The constrained problem \eqref{eq:optimization} can be reformulated mathematically as 
	\begin{equation}\label{eq:optimization1}
		\begin{split}
			\begin{aligned}
				\mathop{\min}_{\mathbf{w}_{1:N}}
				&   \big[f \big(\mathbf{w}_{1:N},\mathbf{y}\big) + \sum_{l=1}^{N} g ( \mathbf{z}_l) 
				\big] \\
				{\mathrm{s.t.}}\ &  \mathbf{z}_l =  \mathbf{\Omega} \mathbf{w}_l, 
				\,l = 1,\ldots,N,
			\end{aligned}
		\end{split}
	\end{equation}
	where the indicator function is defined by
	\begin{equation}\label{eq:g_indicator}
		\begin{split}
			\begin{aligned}
				g(\mathbf{z})= \begin{cases}
					\mathbf{0}, &    \operatorname{\textbf{card}}(\mathbf{z} ) \leq \text{const}, \\
					\infty, &     \text{otherwise}. 
				\end{cases}
			\end{aligned}
		\end{split}
	\end{equation}

	The augmented Lagrangian function associated with \eqref{eq:optimization1} is written as
	\begin{equation}\label{eq:Lagrangian_func}
		\begin{split}
			\begin{aligned}
				&\mathcal{L}(\mathbf{w}_{1:N},\mathbf{z}_{1:N};\bm{\eta}_{1:N}) =   
				f \big(\mathbf{w}_{1:N},\mathbf{y}\big) + \sum_{l=1}^{N} g ( \mathbf{z}_l)   \\ 
				&  + \sum_{l=1}^{N}\bm{\eta}_{l}^\T (\mathbf{z}_{l}- \mathbf{\Omega}\mathbf{w}_{l})
				+ \sum_{l=1}^{N} \frac{\rho}{2} \| \mathbf{z}_{l}- \mathbf{\Omega}\mathbf{w}_{l}\|_2^2 ,
			\end{aligned}
		\end{split}
	\end{equation}
	where $\bm{\eta}_l$ is the Lagrange multiplier, $\rho$ is a penalty parameter, and the vector norm $\left \|\cdot\right \|_2$ is the standard $\ell_2$-norm. Starting at $\mathbf{w}_{l:N} = \mathbf{w}_{l:N}^{(k)}$, $\mathbf{z}_{l:N}= \mathbf{z}_{l:N}^{(k)}$, and $\bm{\eta}_{l:N} = \bm{\eta}_{l:N}^{(k)}$, the iteration steps of ADMM become
	\begin{subequations}
		\begin{align} \label{eq:wl}
			&\mathbf{w}_{1:N}^{(k+1)} = \mathop{\arg\min}_{\mathbf{w}_{1:N}}\mathcal{L}(\mathbf{w}_{1:N},\mathbf{z}_{1:N}^{(k)};\bm{\eta}^{(k)}_{1:N}), \\
			\label{eq:zl}
			&\mathbf{z}_{l}^{(k+1)}= \mathop{\arg\min}_{\mathbf{z}_{l}}\mathcal{L}(\mathbf{w}_{l}^{(k+1)},\mathbf{z}_{l};\bm{\eta}_{l}^{(k)}), \\ 
			\label{eq:etal}
			&\bm{\eta}_{l}^{(k+1)}= \bm{\eta}_{l}^{(k)}+ \rho( \mathbf{z}_{l}^{(k+1)} - \mathbf{\Omega}\mathbf{w}_{l}^{(k+1)}), 
		\end{align}
	\end{subequations}
	which objective is to find a stationary point $(\hat{\mathbf{w}},\hat{\mathbf{z}},\hat{\bm{\eta}})$. The $\mathbf{z}_{l}$ and $\bm{\eta}_{l}$ subproblems compute the updates for each layer. The main benefit of combing with ADMM is that each auxiliary variable $\mathbf{z}_l$ can split the nonconvex sets arising in the weight pruning problems. 
	
	We propose the solution of the $\mathbf{w}_{1:N}$-subproblem. We can write the subproblem in \eqref{eq:wl} as 
	\begin{equation}\label{eq:w-b-subproblem}
		\begin{split}
			\begin{aligned}
				&\mathbf{w}_{1:N}^{(k+1)} =  \mathop{\arg\min}_{\mathbf{w}_{1:N}}   
				f (\mathbf{w}_{1:N},\mathbf{y})   \\
				&\,	+\sum_{l=1}^{N} (\bm{\eta}^{(k)}_{l})^\T \big( \mathbf{z}_{l}^{(k)} - \mathbf{\Omega}\mathbf{w}_{l}\big) + \sum_{l=1}^{N} \frac{\rho}{2} \left\| \mathbf{z}_{l}^{(k)} - \mathbf{\Omega}\mathbf{w}_{l}  \right\|_2^2.
			\end{aligned}
		\end{split}
	\end{equation}
	The function $f$ is the cross-entropy function (see eq.~\eqref{eq:cross-entropy}), and the second term is a square of $L_2$ regularization. Thus, the sum of both functions is differentiable. We then solve \eqref{eq:w-b-subproblem} by computing the $l$-th layer weight
	\begin{equation}\label{eq:w-solution} 
		\begin{split}
			\begin{aligned}
				\mathbf{w}_{l}^{(k+1)} \!\!=\!\! 
				\left\{
				\begin{array}{lr} 
					\!\! \!\!(\frac{\exp^{y^{i}}}{ \sum_{i=1}^n \exp^{y^{i}}} \!-\! \mathbf{I}) \mathbf{y} \!-\!\rho \mathbf{\Omega}   (\mathbf{z}_{l}^{(k)}- \mathbf{\Omega}\mathbf{w}_{l} \!+\! \frac{\bm{\eta}^{(k)}_{l}}{\rho}), \, i = fc,  \\
					\!\!\!\! \frac{\exp^{y^{i}}}{ \sum_{i=1}^n \exp^{y^{i}}} \mathbf{y} - \rho \mathbf{\Omega}  (\mathbf{z}^{(k)}_{l}- \mathbf{\Omega}\mathbf{w}_{l} +\frac{\bm{\eta}_{l}^{(k)}}{\rho}),\, i \neq fc.  
				\end{array}
				\right.
			\end{aligned}
		\end{split}
	\end{equation}

	For solving each $\mathbf{z}_{l}$-subproblem, we write
	\begin{equation}\label{eq:z-v-subproblem}
		\begin{split}
			\begin{aligned}
				\mathbf{z}_{l}^{(k+1)} =\mathop{\arg\min}_{\mathbf{z}_{l}}   
				&g (\mathbf{z}_l) 
				+ \sum_{l=1}^{N} (\bm{\eta}^{(k)}_{l})^\T \big( \mathbf{z}_{l} - \mathbf{\Omega}\mathbf{w}_{l}^{(k+1)}\big) \\
				&\quad + \sum_{l=1}^{N} \frac{\rho}{2} \left\| \mathbf{z}_{l} - \mathbf{\Omega}\mathbf{w}_{l}^{(k+1)}  \right\|_2^2.
			\end{aligned}
		\end{split}
	\end{equation}
	The solution of the $\mathbf{z}_{l}$-subproblem involves keeping the $\text{const}$ largest magnitude elements and zeros out the rest. Inspired by~\cite{Boyd2011admm}, we can obtain
	\begin{equation}\label{eq:z-solution}
		\begin{split}
			\begin{aligned}
				\mathbf{z}_{l}^{(k+1)} = \operatorname{\Pi}(\mathbf{\Omega}\mathbf{w}_{l}^{(k+1)}  - {\bm{\eta}^{(k+1)}_{l}}/{\rho}),
			\end{aligned}
		\end{split}
	\end{equation}
	where $\Pi(\cdot)$ is projection operator onto the constraint set $\{ \mathbf{x} | \operatorname{\textbf{card}}(\mathbf{x}) \leq \text{const} \}$. Note that the parameter $\text{const} $ controls the trade-off between least-squares errors and desired carnality. 
	
	After all the iteration, we can obtain the stationary point $(\hat{\mathbf{w}},\hat{\mathbf{z}},\hat{\bm{\eta}})$. 
	Here, we can use other splitting methods such as ALS~\cite{Goldstein2009Split} and PRS~\cite{He2014peaceman} to compute the objective function \eqref{eq:optimization1}. When the objective function~\eqref{eq:Lagrangian_func} is convex, then the function globally converges to the global point $(\hat{\mathbf{w}},\hat{\mathbf{z}},\hat{\bm{\eta}})$~\cite{Boyd2011admm}. However, the function here is nonconvex. We need strict assumption conditions that ensure the convergence results, for example, proper choices of $\rho$ and $\mathbf{\Omega}$~\cite{Gao2020Variable_new}. In the following, we will prove the convergence results.

	\subsection{Convergence Analysis}
	
	Since pruning a deep neural network is nonconvex problem, even when both the loss function and constraints are convex, we will consider the nonconvexity in this article. Before that, we make the following definition.  
	
	\begin{definition}[Prox-regularity \cite{Rockafellar1998Variational,Poliquin1996Prox}]
		\label{lemma:pro-regularity}
		If $f(\mathbf{x},\mathbf{y})$ is is  prox-regular with a positive constant $M$, then for any $\mathbf{x}_1$, $\mathbf{x}_2 $ in a neighborhood of $\mathbf{x}$, there exists $M> 0$ such that  
		\begin{equation}\label{eq:prox-regular}
			\begin{split}
				\begin{aligned}
					f(\mathbf{x}_1,\mathbf{y}) &- f(\mathbf{x}_2,\mathbf{y})   \\
					&\geq- \frac{M}{2} \|\mathbf{x}_1 - \mathbf{x}_2 \|^2  +  \langle \partial_{\mathbf{x}} f(\mathbf{x},\mathbf{y}), \mathbf{x}_1 - \mathbf{x}_2\rangle.
				\end{aligned}
			\end{split}
		\end{equation}
	\end{definition}
	
	We then make the following assumptions for establishing the convergence.
	\begin{assumption}
		\label{assump:prox-regular}
		The function $f(\mathbf{w}_{1:N},\mathbf{y})$ is prox-regular at $\mathbf{w}_{1:N}$ with constants $M_w$.	
	\end{assumption}
	
	\begin{assumption}
		\label{assump:omega}
		$\mathbf{\Omega}$ is full-column rank with 
		\[\mathbf{\Omega} \mathbf{\Omega}^\T \succeq \kappa^2 \mathbf{I}.\]
	\end{assumption}	
	
	Assumption \ref{assump:prox-regular} can be used to bound the partial of the loss function $f(\mathbf{w}_{1:N},\mathbf{y})$. According to Assumption \ref{assump:prox-regular}, we have the following equations 
	\begin{equation}\label{eq:prox-regular-case}
		\begin{split}
			\begin{aligned}
				&f(\mathbf{w}_{1:N}^{(k+1)},\mathbf{y}) - f(\mathbf{w}_{1:N}^{(k)},\mathbf{y}) 
				\geq -\frac{M_w}{2} \left\|\mathbf{w}_{1:N}^{(k+1)} - \mathbf{w}_{1:N}^{(k)} \right\|^2 \\ & \quad +  \langle \partial_{\mathbf{w}} f(\mathbf{w}_{1:N}^{(k+1)},\mathbf{y}), \mathbf{w}_{1:N}^{(k+1)} - \mathbf{w}_{1:N}^{(k)}\rangle.\\	
			\end{aligned}
		\end{split}
	\end{equation}
	Due to the optimality condition on $\mathcal{L}(\mathbf{w}_{1:N},\mathbf{z}_{1:N};\bm{\eta}_{1:N})$,we can obtain
	\begin{equation}
		\label{eq:gradient_lang}
		\begin{split}
			\begin{aligned}
				&\partial_{\mathbf{w}}\mathcal{L}(\mathbf{w}_{1:N}^{(k+1)},\mathbf{z}_{1:N}^{(k)};\bm{\eta}_{1:N}^{(k)}) = 0. \\
			\end{aligned}
		\end{split} 
	\end{equation}
	which implies that 
	\begin{equation}
		\label{eq:first_order_x}
		\begin{split}
			\begin{aligned}
				\partial f_{\mathbf{w}}(\mathbf{w}_{1:N}^{(k+1)}, \mathbf{y}) =	\mathbf{\Omega}_{1:N}^\T	\bm{\eta}_{1:N}^{(k)}	+ \rho \mathbf{\Omega}_{1:N}^\T ( \mathbf{z}_{1:N}^{(k)} - \mathbf{\Omega}_{1:N} \mathbf{w}_{1:N}^{(k+1)}),
			\end{aligned}
		\end{split} 
	\end{equation}
	where $\mathbf{\Omega}_{1:N} =  \operatorname{blkdiag}(\mathbf{\Omega},\mathbf{\Omega},\ldots,\mathbf{\Omega})$ with the block diagonal matrix operator $\operatorname{blkdiag}(\cdot)$. Now, we are ready to proof that the sequence $\mathcal{L}(\mathbf{w}_{1:N}^{(k)},\mathbf{z}_{1:N}^{(k)};\bm{\eta}_{1:N}^{(k)}) $ is monotonically nonincreasing when updating the variables $\mathbf{w}_{1:N}$.
	
	\begin{lemma}
		\label{lemma:wb-sub}
		Let Assumptions \ref{assump:prox-regular} and \ref{assump:omega} be satisfied. Also let $\{\mathbf{w}^{(k)}_{1:N},\mathbf{z}^{(k)}_{1:N};\bm{\eta}^{(k)}_{1:N}\}$ be the iterative sequence. If $\rho \geq \frac{ M_{w}}{\kappa^2} $, then the sequence $\mathcal{L}(\mathbf{w}^{(k)}_{1:N},\mathbf{z}^{(k)}_{1:N};\bm{\eta}^{(k)}_{1:N}) $ is nonincreasing with the updates of $\mathbf{w}_{1:N}$.
	\end{lemma}
	
	\begin{proof}
		See Appendix \ref{sec:proof}.
	\end{proof}
	\begin{lemma}
		\label{lemma:nonincreasing}
		
		Let assumptions of Lemma \ref{lemma:wb-sub} be satisfied. The sequence $\mathcal{L}(\mathbf{w}^{(k)}_{1:N},\mathbf{z}^{(k)}_{1:N};\bm{\eta}^{(k)}_{1:N})$ is monotonically nonincreasing. 
	\end{lemma}
	
	\begin{proof}
		For the $\bm{\eta}_{1:N}$-subproblem, we obtain 
		\begin{equation}
			\label{eq:zeta_sub}
			\begin{split}
				\begin{aligned}
					& \mathcal{L}(\mathbf{w}^{(k+1)}_{1:N},\mathbf{z}^{(k+1)}_{1:N};\bm{\eta}^{(k+1)}_{1:N})
					- 
					\mathcal{L}(\mathbf{w}^{(k+1)}_{1:N},\mathbf{z}^{(k+1)}_{1:N};\bm{\eta}^{(k)}_{1:N}) \\
					&  =  
					\langle	\bm{\eta}_{1:N}^{(k+1)}-\bm{\eta}_{1:N}^{(k)}, \mathbf{z}_{1:N}^{(k+1)} -  \mathbf{\Omega}_{1:N} \mathbf{w}_{1:N}^{(k+1)}  \rangle \\
					&=  \frac{1}{\rho} \left\| \bm{\eta}_{1:N}^{(k+1)} -  \bm{\eta}_{1:N}^{(k)} \right\|^2.
				\end{aligned}
			\end{split} 
		\end{equation}

		Using Lemma \ref{lemma:wb-sub}, we have 
		\begin{equation} \label{eq:non_increasing}
			\begin{split}
				\begin{aligned}
					& \mathcal{L}(\mathbf{w}^{(k+1)}_{1:N},\mathbf{z}^{(k+1)}_{1:N};\bm{\eta}^{(k+1)}_{1:N})
					- 
					\mathcal{L}(\mathbf{w}^{(k)}_{1:N},\mathbf{z}^{(k)}_{1:N};\bm{\eta}^{(k)}_{1:N}) \\
					&\geq  \!  \frac{ M_{w} \!  \! - \!  \! \rho \kappa^2 }{2} 
					\left\| \mathbf{w}_{1:N}^{(k+1)}  \!- \! \mathbf{w}_{1:N}^{(k)} \right\|^2   
					+  \frac{1}{\rho} \left\| \bm{\eta}_{1:N}^{(k+1)}  \!- \!  \bm{\eta}_{1:N}^{(k)} \right\|^2. 
				\end{aligned}
			\end{split} 
		\end{equation}
		Since the condition $\rho \geq \frac{ M_{w}}{\kappa^2} $ is satisfied, we write
		\begin{equation} 
			\begin{split}
				\begin{aligned}
					\mathcal{L}(\mathbf{w}^{(k+1)}_{1:N},\mathbf{z}^{(k+1)}_{1:N};\bm{\eta}^{(k+1)}_{1:N})
					- 
					\mathcal{L}(\mathbf{w}^{(k)}_{1:N},\mathbf{z}^{(k)}_{1:N};\bm{\eta}^{(k)}_{1:N})  \geq 0. 
				\end{aligned}
			\end{split} 
		\end{equation}
	\end{proof}

	Based on Lemma \ref{lemma:nonincreasing}, we can now establish the convergence results in the following theorem.

	\begin{theorem}
		\label{theorem:ADMM-type}
		Let assumptions of Lemmas~\ref{lemma:nonincreasing} be satisfied.
		Then, the sequence $\{\mathbf{w}_{1:N}^{(k)}, \mathbf{z}_{1:N}^{(k)}, \bm{\eta}_{1:N}^{(k)} \}$ generated by the  converges to a local minimum $(\hat{\mathbf{w}},\hat{\mathbf{z}},\hat{\bm{\eta}})$. 
	\end{theorem}
	
	\begin{proof}
		By Lemmas~\ref{lemma:wb-sub} and \ref{lemma:nonincreasing}, the sequence $\mathcal{L}(\mathbf{w}^{(k)}_{1:N},\mathbf{z}^{(k)}_{1:N};$$\bm{\eta}^{(k)}_{1:N})$ is monotonically nonincreasing. Since  $\mathcal{L}(\mathbf{w}^{(k)}_{1:N},\mathbf{z}^{(k)}_{1:N};\bm{\eta}^{(k)}_{1:N})$ is upper bounded by $\mathcal{L}(\mathbf{w}^{(0)}_{1:N},$$\mathbf{z}^{(0)}_{1:N};\bm{\eta}^{(0)}_{1:N})$, and lower bounded by $f(\mathbf{w}_{1:N}^{(k)},\mathbf{y})$. 
		We assume there exists a local minimum $\hat{\mathbf{w}}$ such that the sequence $\{\mathbf{w}_{1:N}^{(k)}\}$ converges to  $\hat{\mathbf{w}}$, which is a local minimum of $\mathbf{w}_{1:N}$-subproblem. We deduce the iterative sequence $\{\mathbf{w}_{1:N}^{(k)},\mathbf{z}_{1:N}^{(k)},\bm{\eta}_{1:N}^{(k)}\}$ generated by the algorithm is locally convergent to $(\hat{\mathbf{w}},\hat{\mathbf{z}},\hat{\bm{\eta}})$.
	\end{proof}

	\subsection{Summary}
	
	To explicitly utilize the graph structure of SSGCNet, we transform the input signals into the frequency domain, and propose the WNFG method to represent the signals. Note that the EEG signal segment should be short enough to make that the signal is stationary within the segment. The duration of the epileptiform wave is the cause of the instability of the signal, so the epileptiform wave cannot change significantly during a segment. The length of a segment should be less than the duration of an epileptic seizure waveform, which will be discussed in Section \ref{sec:settings}. 
	
	Based on the WNFG method, we propose the SSGCNet model, including node aggregation, node sequential convolution and fully connected layer. Due to the requirement of large-scale datasets, we develop a lightweight version of the model by using the ADMM weight pruning method. At each epoch $k_e$, we execute the training and pruning operations. Our SSGCNet model offers competitive accuracy with large speed and memory savings. Now, the main steps of our model are summarized in Algorithm~\ref{alg:SSGCNet}.
	
	\begin{algorithm}[htb] \label{alg:SSGCNet}
		\caption{SSGCNet} 
		\KwIn{weights $\mathbf{w}_{1:N}$, signal $\mathbf{x}$, auxiliary variables $\mathbf{z}_{1:N}$ and $\bm{\eta}_{1:N}$, parameter $\rho$, parameters $\bm{\theta}_{1:n}$} 
		\KwOut{$\hat{\mathbf{w}}$, $\hat{\mathbf{y}}$.}     
		compute graph representation $\mathbf{y}$ by \eqref{eq:x_i} and \eqref{eq:alpha}\;    		
		\For{\text{epoch} $k_e \leq K_{e_{\max}}$  }
		{
			compute the node aggregation $\{o^i\}$ by \eqref{eq:aggregation}\;
			compute the node sequential convolution layer $\{y^{i}\}$ by \eqref{eq:convolution}\;
			\While{not convergent} 
			{ 
				compute $\mathbf{w}_{1:N}$ by \eqref{eq:w-solution}\;  
				compute $\mathbf{z}_{1:N}$ by \eqref{eq:z-solution}\;
				compute $\bm{\eta}_{1:N}$ by \eqref{eq:etal}\;
			}
			compute the fully connected layer $\{y^{fc}\}$ by \eqref{eq:y_fc}\;
			\While{not convergent} 
			{ 
				compute $\mathbf{w}_{1:N}$ by \eqref{eq:w-solution}\;  
				compute $\mathbf{z}_{1:N}$ by \eqref{eq:z-solution}\;
				compute $\bm{\eta}_{1:N}$ by \eqref{eq:etal}\;
			}
			return $\hat{\mathbf{y}} = [y^{fc}_1,\ldots,y^{fc}_{n_y}]$ and $\hat{\mathbf{w}}$\; 
		}
	\end{algorithm}

	\section{Experiments and Discussions}
	\label{sec:results}
	
	In this section, we evaluate our method on the Bonn dataset and the SSW dataset~\cite{andrzejak2001indications}. For comparison, we also report the experimental results on various existing deep learning models, including Multi-Layer Perceptron (MLP)~\cite{ruck1990multilayer}, Graph Neural Network (GNN)~\cite{Xu2019How}, Convolutional Neural Network (CNN)~\cite{sors2018convolutional}.

	\subsection{Datasets}
	
	In this article, we evaluate our method on two single-channel epilepsy EEG signal datasets, namely the Bonn dataset and the SSW dataset.

	The Bonn dataset is a single-channel epilepsy EEG signal dataset containing five categories, including Subset A, Subset B, Subset C, Subset D, and Subset E. Subset A and Subset B are from $5$ healthy subjects. Subsets A and B are collected from the normal EEG signal of the subject with eyes open and closed, respectively. Subsets C, D, and E are from $5$ patients with confirmed epilepsy, including reverse area, epilepsy lesion area and epileptic seizures. Each category is composed of $100$ EEG signals containing $4097$ data points. The sampling frequency is $173.6$ Hz, and the total duration of each signal is $23.6$ seconds. The dataset has been preprocessed, such as myoelectricity and power frequency interference. For a more detailed introduction, please refer to~\cite{andrzejak2001indications}.
	
	The SSW dataset is a single-channel absence epilepsy EEG signal dataset containing two categories (seizure data, non-seizure data). Each category is composed of $10473$ EEG signals containing $200$ data points. The sampling frequency is $200$ Hz. The total duration is $1$ second. Data are collected from $10$ patients diagnosed with absence epilepsy. Absence epileptic seizures have no obvious symptoms of convulsions. The clinical diagnosis is mainly based on the spikes and slow waves (SSW) appearing in the EEG signal. Each seizure EEG signal contains at least $1$ complete spike and slow waves data. Non-seizure data was also captured from $10$ patients with their non-seizure stage.

	\begin{figure*}[t]
		\setlength{\abovecaptionskip}{-0cm}
		\setlength{\belowcaptionskip}{-0cm}
		\centering \includegraphics[width=\textwidth]{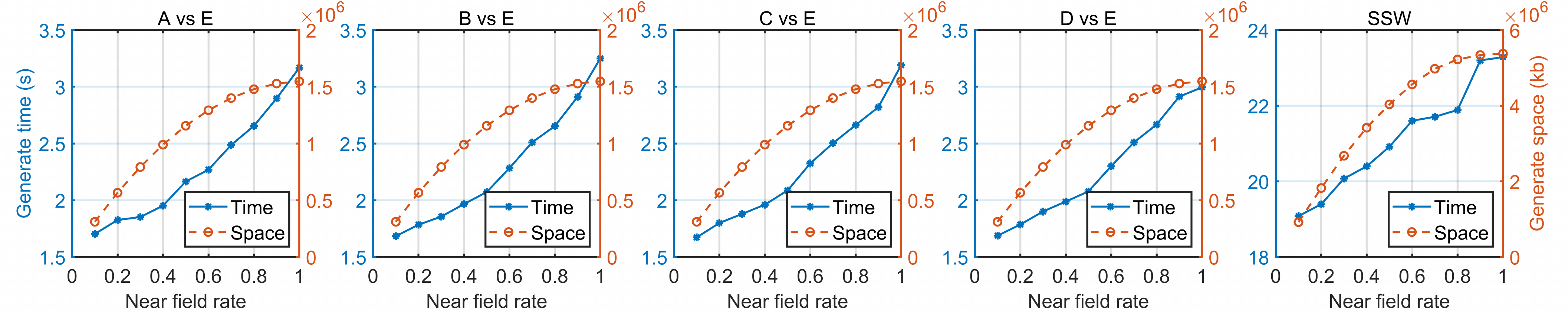}
		\caption{The time and space overhead of WNFG in different near field rates.}
		\label{time_space}
	\end{figure*}

	\begin{figure*}[!htbp]
		\setlength{\abovecaptionskip}{-0cm}
		\setlength{\belowcaptionskip}{-0cm}
		\centering \includegraphics[width=\textwidth]{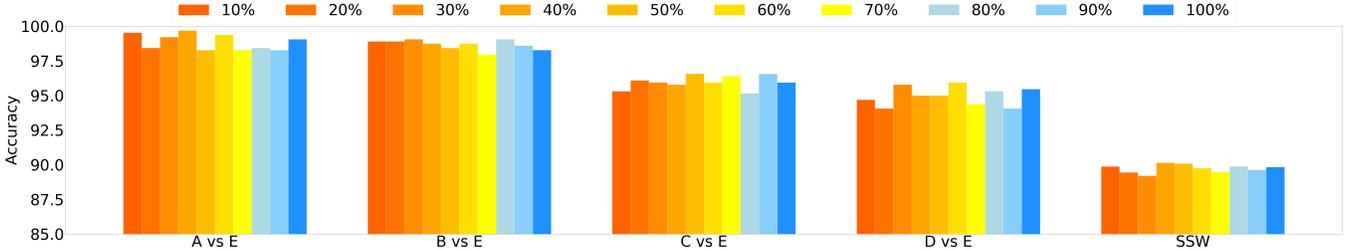}
		\caption{Accuracy of different near field rates of WNFG in Bonn and SSW dataset. We use MLP as the baseline deep learning model.}
		\label{WNFG_acc}
	\end{figure*}
	
	\begin{figure*}[!htbp]
		\setlength{\abovecaptionskip}{-0cm}
		\setlength{\belowcaptionskip}{-0cm}
		\centering \includegraphics[width=0.98\textwidth]{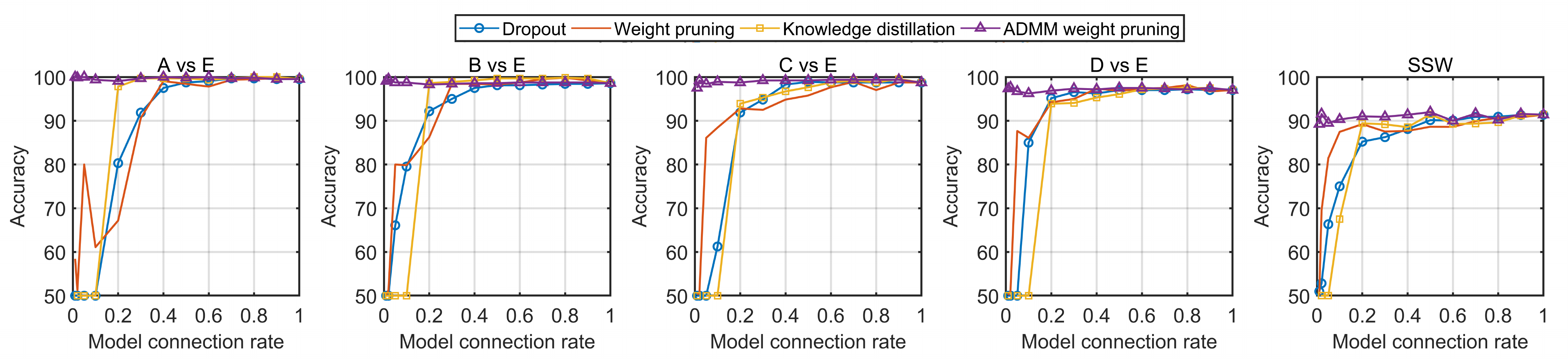}
		\caption{Accuracy of the four pruning methods in Bonn and SSW dataset. We use MLP as the baseline model.}
		\label{experiment2}
	\end{figure*}

	\begin{figure}[!htbp]
		\setlength{\abovecaptionskip}{-0cm}
		\setlength{\belowcaptionskip}{-0cm}
		\centering \includegraphics[width=0.45\textwidth]{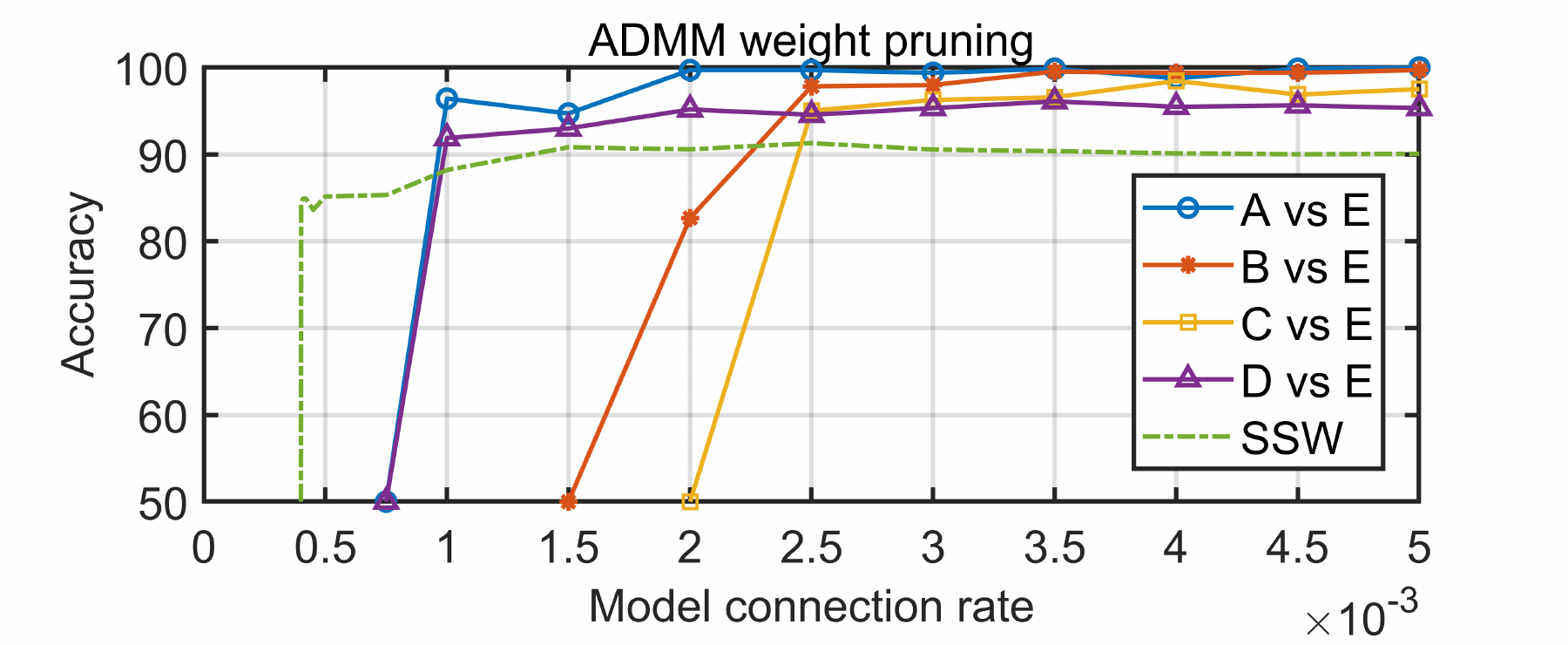}
		\caption{The accuracy of the ADMM weight pruning method in the Bonn and SSW datasets.}
		\label{fig:exp2_admm}
	\end{figure}
	
	\subsection{Experimental Settings}
	\label{sec:settings}
	
	We implemented the experiment in PyTorch, which was tested on 4 Nvidia Titan XP GPUs. All the methods in the following have the same experimental settings.

	For the EEG signal segment, the duration of epileptiform waves is around $3$ seconds to several hours~\cite{glauser2010ethosuximide}. Meanwhile, the EEG signal segment includes enough epileptiform vibration cycles, and the frequency can only be resolved by repeating enough vibration cycles. The fundamental frequency of epilepsy is around $2.5$—$4$ Hz~\cite{crunelli2002childhood}, which is $0.25$ seconds and $0.4$ seconds when converted into cycles. Since a segment contains multiple cycles, it is generally at least greater than $0.4$ seconds. In this article, we segment the single EEG signal from Bonn dataset into small segments of equal length with around $1.5$ seconds (containing $256$ data points without overlap). For the SSW dataset, we set $1$ second as each segment.

	We use the same deep learning model to compare with existing pruning methods. We selected the MLP model as the baseline model for weight pruning optimization. The model contains $4$ layers, consisting of $1$ input layer, $2$ hidden layers, and $1$ output layer. The maximum epoch $K_{e_{\max}}$ is $50$. The input size is $[n,\,200]$, where $n$ is the length of the signal, and the number of hidden layer neurons is $200$. The size of the hidden layer remains unchanged as $[200,\,200]$, and the size of the output layer is $[200,\,c]$ with the number of categories $c$.
	
	We compare the proposed approach with the other three pruning optimization methods, including the dropout method~\cite{baldi2013understanding}, the weighted pruning method~\cite{zhang2018systematic}, and the knowledge distillation method~\cite{yim2017gift}. The dropout method randomly resets the weights to zero, thereby reducing the number of model connections. The weight pruning method adopts the removal of smaller weights, thereby reducing the overfitting of the model. The knowledge distillation method improves the classification performance of the small model by using the small-scaled student model to learn the features extracted by the teacher model.

	This article also uses four different deep learning models to verify the versatility of our methods on different deep learning models. We test four deep learning models used in EEG signal classification tasks: MLP, GNN, CNN, and SSGCNet. 
	MLP is a multilayer perceptron that includes $5$ fully connected layers. The numbers of the neuron are $[n, \,128]$, $[128,\, 64]$, $[64,\,32]$, $[32,\, 16]$, and $[16,\, 2]$, respectively. Here, $n$ is the node number of the input graph.
	GNN is a graph neural network that includes node aggregation and fully connected layers. The model includes two-hop node aggregation and $5$ fully connected layers, and other parameter settings of the model are the same as those of MLP.
	CNN is a deep learning model that includes the one-dimensional convolutional layer, the max-pooling layer and the fully connected layer. The convolutional layer has $4$ layers, the max-pooling layer has 4 layers, and the fully connected layer has 2 layers. The convolution kernel size is $[1,\,3]$, its step size is 1, and the number of channels in each layer is $8$, $16$, $32$, and $64$, respectively. The size of the largest pooling layer is $[1,\,2]$, and the step size is $2$. The dimensions of the fully connected layer are $[64,\, 16]$ and $[16,\, 2]$.
	SSGCNet includes node aggregation, one-dimensional convolution, and fully connected layers (see Fig. \ref{fig:network} for details). In the one-dimensional convolutional layer, the size of the convolution kernel is $[1,\,3]$, its step size is 1, and the number of channels in each layer is $8$, $16$, $32$, and $64$. The size of the largest pooling layer is $[1,\,2]$, and the step size is $2$. Other settings are the same as those of GNN.

	\subsection{Performance Comparison of EEG Signal Classification}
	
	\subsubsection{Time and space complexity}

	We decompose signals into $10$ segments according to the different near field rates of the connection rules ($n = 10$). The near field rate represents the proportion of the connection rule. For example, the near field rate is $0.1$, and the node is only connected to other nodes with a maximum distance of $10\%$ in the graph. For the Bonn and SSW datasets, with different near field rates, we test the generation time and space occupancy rate, and evaluate the accuracy of different near field rate datasets.
	
	It can be seen from Fig.~\ref{time_space} that the generation time and space occupied by different datasets are reduced. In the Bonn dataset, the near field rate is from $1$ to $0.1$. The time cost is reduced by $1.45$ seconds. The space occupancy rate is reduced by $5$ times. In the SSW dataset, the near field rate from $1$ to $0.1$ reduces the time overhead by around $4$ seconds and the space occupancy rate by $5.8$ times. This demonstrates the effectiveness of our WNFG representation.
	
	As shown in Fig.~\ref{WNFG_acc}, when the near field rate decreases, the classification performance under the same deep learning model does not decrease significantly. For example, the classification performance is even higher when the decrease of the near field rate on A vs E is lower. The experiment demonstrates that reducing the number of connections in the graph representation does not cause a significant drop in the performance of the model. Hence, in classifying epilepsy EEG signals, our WNFG can significantly reduce the time and space calculation and retain sufficient classification information.
	
	\subsubsection{Different weight pruning strategies} 
	To compare the pruning methods, we evaluate different pruning optimization methods on the same deep learning model. We also use the WNFG method to analyze all deep learning pruning strategies. Our evaluation method divides the model connection rate of each deep learning model optimization strategy into $10$ segments for evaluation. The model connection rate represents the proportion of the original model containing the number of non-zero parameters. We test different compression ratios on the model structure. In particular, deep learning dropout is used as a control group to compare with other pruning optimization methods.

	\begin{figure*}[!htbp]
		\setlength{\abovecaptionskip}{-0cm}
		\setlength{\belowcaptionskip}{-0cm}
		\centering \includegraphics[width=\textwidth]{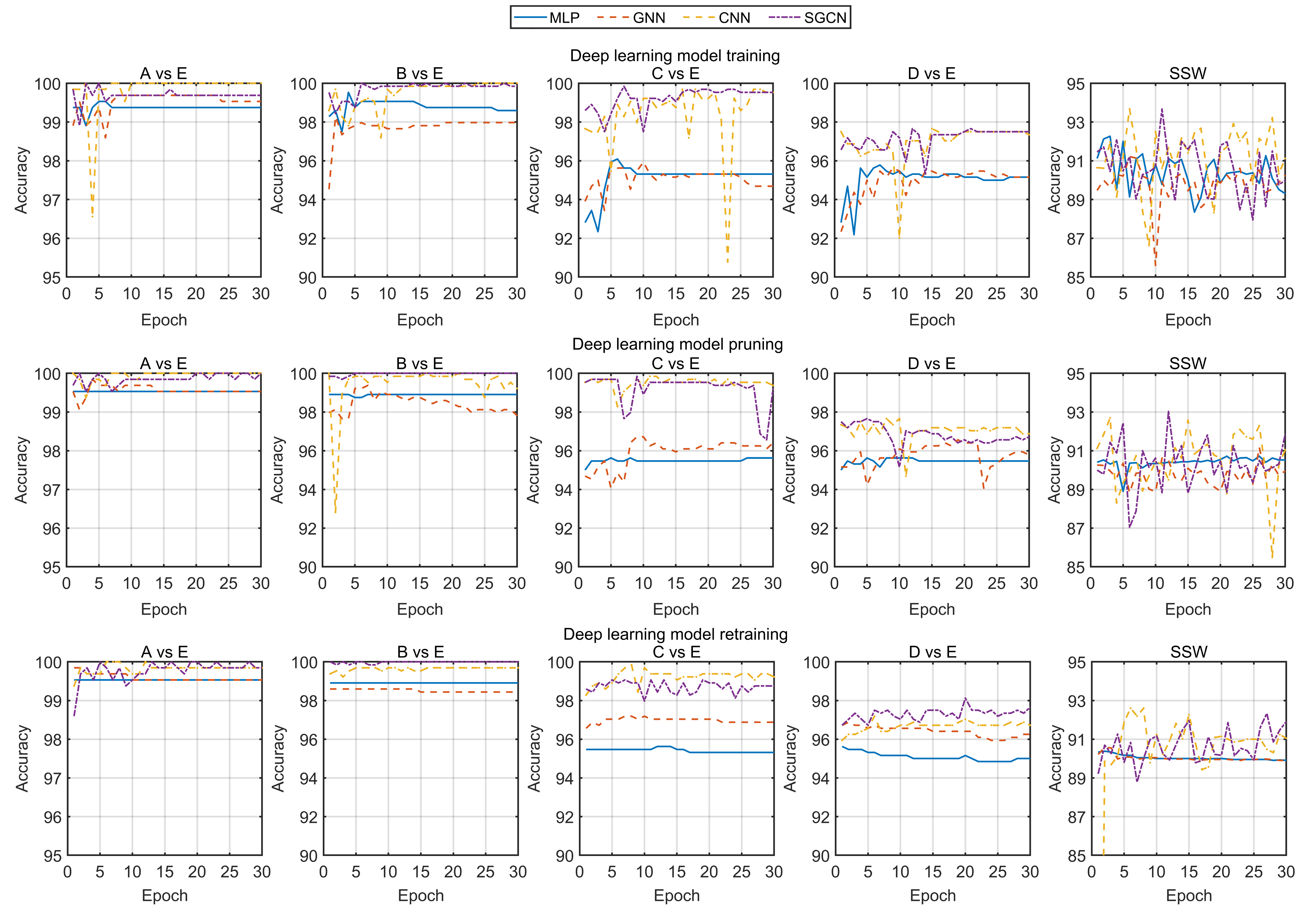}
		\caption{Training curves of four deep learning models with ADMM.}
		\label{model_ADMM}
	\end{figure*}
	
	Fig.~\ref{experiment2} shows the classification results of our method and other methods on the deep learning model.
	When the model connection rate is low, all pruning optimization methods can maintain a high classification performance. As the model connection rate decreases, the performance of different pruning methods changes. Our weight pruning classification results are significantly better than other models. Meanwhile, our method maintains higher accuracy under the same compression ratio. When the model connection rate reaches $0.01$ ($100$ times), other methods have loss the accuracy shown in  Fig.~\ref{experiment2}. Our proposed method can still keep good classification performance.

	\begin{table}[!htbp]
		\setlength\tabcolsep{4 pt}
		\caption{The number of parameters on MLP-ADMM.}
		\setlength{\abovecaptionskip}{-0cm}
		\setlength{\belowcaptionskip}{-0cm}
		\label{MLP_ADMM}
		\centering
		\begin{tabular}{c|c|c|c|c|c|c|c}
			\hline
			\hline
			\multicolumn{1}{c|}{Layers} & \multicolumn{7}{c}{MLP} \\
			\cline{2-8}
			&  \multicolumn{1}{c}{-} & \multicolumn{1}{c}{A vs E} & \multicolumn{1}{c}{B vs E} & \multicolumn{1}{c}{C vs E} & \multicolumn{1}{c|}{D vs E} & \multicolumn{1}{c}{-} & \multicolumn{1}{c}{SSW}\\
			\hline
			\hline
			Layer 1 & 51200   & 52      & 103     & 128     & 52      & 40000  & 17    \\			
			Layer 2 & 40000   & 40      & 80      & 100     & 40      & 40000  & 17    \\
			Layer 3 & 40000   & 40      & 80      & 100     & 40      & 40000  & 17    \\
			Layer 4 & 400     & 1       & 1       & 1       & 1       & 400    & 1     \\
			Non-train & 602     & 602     & 602     & 602     & 602     & 602    & 602    \\
			\hline
			\hline
			Total   & 132202  & 735     & 866     & 931     & 735     & 121002 & 654    \\
			\hline
			\hline		
		\end{tabular}
	\end{table}

	\begin{table*}[!htbp]
		\setlength\tabcolsep{2 pt}
		\caption{The performance of different methods.}
		\setlength{\abovecaptionskip}{-0cm}
		\setlength{\belowcaptionskip}{-0cm}
		\label{SOTA_ADMM}
		\centering
		\resizebox{\textwidth}{!}{
			\begin{tabular}{c|c|c|c|ccc|ccc|ccc|ccc|ccc}
				\hline
				\hline
				\multicolumn{1}{c|}{} & \multicolumn{1}{c|}{Complexity} & \multicolumn{1}{c|}{Number} & \multicolumn{1}{c|}{Operation} & \multicolumn{12}{c|}{Bonn} & \multicolumn{3}{c}{SSW}\\
				\cline{5-19}
				Methods & of & of & of & \multicolumn{3}{c|}{A vs.E} & \multicolumn{3}{c|}{B vs.E} & \multicolumn{3}{c|}{C vs.E} & \multicolumn{3}{c|}{D vs.E} & & \\
				\cline{5-19}
				& Graph & Parameters & multiplication &  Acc & Spe & Sen  & Acc & Spe & Sen & Acc & Spe & Sen & Acc & Spe & Sen & Acc & Spe & Sen \\
				\hline
				\hline
				VG-GNN~\cite{lacasa2008time51} & $O(n^3)$ &44306& 43552& 0.997 & 0.999 & 0.994 & 0.975 & 0.987 & 0.962 & 0.963 & 0.969 & 0.956 & 0.916 & 0.900 & 0.931 & 0.837 & 0.724 & \textbf{0.949}  \\
				LPVG-GNN~\cite{zhou2012limited243} & $O(n^2)$ &44306& 43552& 0.978 & 0.962 & 0.994 & 0.969 & 0.956 & 0.981 & 0.963 & \textbf{0.975} & 0.950 & 0.906 & 0.856 & 0.956 & 0.864 & 0.818 & 0.910  \\
				HVG-GNN~\cite{luque2009horizontal59} & $O(n\log n)$ &44306& 43552& 0.988 & 0.987 & 0.988 & 0.916 & 0.912 & 0.919 & 0.972 & 0.950 & 0.994 & 0.925 & \textbf{0.969} & 0.881 & 0.873 & 0.847 & 0.898  \\
				LPHVG-GNN~\cite{gao2016multiscale245} & $O(n\log n)$ &44306& 43552& 0.994 & 0.994 & 0.994 & 0.959 & 0.950 & 0.969 & 0.966 & 0.969 & 0.962 & 0.922 & 0.940 & 0.922 & 0.862 & 0.837 & 0.887  \\
				WNFG-GNN-ADMM  & $O(n)$ &9467& 8713& \textbf{1.000} & \textbf{1.000} & \textbf{1.000} & \textbf{0.997} & \textbf{0.997} & \textbf{0.997} & \textbf{0.984} & 0.969 & \textbf{1.000} & \textbf{0.972} & 0.963 & \textbf{0.981} & \textbf{0.914} & \textbf{0.975} & 0.853  \\
				\hline
				\hline
				WOG-2DCNN~\cite{wang2020weighted} & $O(n^2)$ & 650210 & 278429824 & \textbf{1.000} & \textbf{1.000} & \textbf{1.000} & 0.995 & \textbf{0.997} & 0.994 & \textbf{0.995} & \textbf{0.994} & \textbf{0.997} & \textbf{0.978} & \textbf{0.981} & \textbf{0.975} & 0.895 & \textbf{0.986} & 0.805  \\			
				WOG-2DCNN-ADMM   & $O(n^2)$ & 6657 & 2817149& \textbf{1.000} & \textbf{1.000} & \textbf{1.000} & \textbf{0.997} & \textbf{0.997} & \textbf{0.997} & 0.992 & 0.988 & \textbf{0.997} & 0.967 & 0.963 & 0.972 & 0.909 & 0.975 & \textbf{0.842}  \\ %65163			
				WNFG-2DCNN-ADMM  & $O(n)$ & 6657 & 2817149& 0.998 & 0.997 & \textbf{1.000} & \textbf{0.997} & \textbf{0.997} & \textbf{0.997} & 0.992 & 0.991 & 0.994 & 0.967 & 0.963 & 0.972 & \textbf{0.910} & 0.978 & \textbf{0.842}  \\
				\hline
				\hline
				LPVG-SGCN~\cite{wang2020sequential} & $O(n^2)$ & 47282 & 678016& 0.992 & \textbf{1.000} & 0.984 & 0.988 & 0.997 & 0.978 & 0.966 & 0.969 & 0.963 & 0.939 & 0.947 & 0.931 & 0.802 & 0.759 & 0.843  \\
				LPVG-SGCN-ADMM & $O(n^2)$ & 4869 & 68058& 0.991 & \textbf{1.000} & 0.981 & 0.986 & 0.984 & 0.988 & 0.964 & 0.975 & 0.953 & 0.956 & 0.963 & 0.950 & 0.768 & 0.693 & 0.844  \\
				WOG-SGCN-ADMM & $O(n^2)$ & 4869 & 68058& \textbf{1.000} & \textbf{1.000} & \textbf{1.000} & \textbf{1.000} & \textbf{1.000} & \textbf{1.000} & \textbf{0.988} & 0.975 & \textbf{1.000} & 0.969 & \textbf{0.984} & 0.953 & 0.911 & \textbf{0.989} & 0.832  \\
				Ours method & $O(n)$ & 4869 & 68058& 0.998 & 0.997 & \textbf{1.000} & \textbf{1.000} & \textbf{1.000} & \textbf{1.000} & \textbf{0.988} & \textbf{0.984} & 0.991 & \textbf{0.977} & 0.972 & \textbf{0.981} & \textbf{0.919} & 0.986 & \textbf{0.852}  \\
				\hline
				\hline
			\end{tabular}
		}
	\end{table*}

	\begin{table}[!htbp]
		\setlength\tabcolsep{5 pt}
		\caption{The number of parameters on differnet deep learning models.}
		\setlength{\abovecaptionskip}{-0cm}
		\setlength{\belowcaptionskip}{-0cm}
		\label{Model_ADMM}
		\centering
		\begin{tabular}{c|cc|cc|cc}
			\hline
			\hline
			\multicolumn{1}{c|}{Layers} & \multicolumn{6}{c}{Deep learning models} \\
			\cline{2-7}
			of &  \multicolumn{2}{c|}{GNN} & \multicolumn{2}{c|}{2DCNN} & \multicolumn{2}{c}{SSGCNet} \\
			\cline{2-7}
			model& - & ADMM & - & ADMM & - & ADMM \\
			\hline
			\hline
			Layer 1 & 32768   & 6554    & 72      & 1       & 24      & 3       \\			
			Layer 2 & 8192    & 1639    & 1152    & 12      & 384     & 39      \\
			Layer 3 & 2048    & 410     & 41472   & 415     & 4608    & 461     \\
			Layer 4 & 512     & 103     & 82944   & 830     & 9216    & 922     \\
			Layer 5 & 32      & 7       & 524288  & 5243    & 32768   & 3277    \\
			Layer 6 & -       & -       & 128     & 2       & 128     & 13      \\
			Non-train  & 754     & 754     & 154     & 154     & 154     & 154     \\
			\hline
			\hline
			Total   & 44306   & 9467    & 650210  & 6657    & 47282   & 4869    \\
			\hline
			\hline		
		\end{tabular}
	\end{table}

	In order to further verify the performance of the proposed method, we use the ADMM weight pruning method to reduce the model connection rate. As can be seen from Fig.~\ref{fig:exp2_admm}, the model connection rate of the ADMM weight pruning method can be reduced to $0.001$ in the A vs E, which achieves up to $1000$ times of redundant edge reduction. Similarly, on the SSW dataset, the model connection rate can be reduced to $0.002$ ($500$ times) in the B vs E, and the model connection rate can be reduced to $0.001$ ($1000$ times) in the D vs E. 
	
	The minimum parameters of the MLP model of ADMM weight pruning are shown in Table~\ref{MLP_ADMM}. Our ADMM weight pruning method can maintain accuracy without loss in almost all model connection rates. We achieve a $10$-fold weight reduction without accuracy loss, and the classification accuracy is higher than other methods. Other methods achieve a compression rate of $2$ times, but the accuracy will be sharply reduced. In the Knowledge distillation method, the performance can remain stable when the model connection rate is $0.2$, but the classification performance will completely fail when the model connection rate reaches $0.1$.  
	
	The results in Figs.~\ref{experiment2} and \ref{fig:exp2_admm} show that our weighted pruning method can maintain high accuracy when extremely low compression ratio. Our method compresses the model to the smaller scale with high accuracy than the compared methods. Those experiments demonstrate that the ADMM weight pruning method is computationally superior to the existing literature.
	
	\subsubsection{Different deep learning models} 
	
	Here, we verify the universal applicability of ADMM-type splitting and weight pruning strategy in deep learning models. For the Bonn dataset and the SSW dataset, we select four different deep learning models, including MLP, GNN, CNN and SSGCNet. We train the original model on both datasets, then use ADMM weight pruning to optimize the model and retrain the pruned model. In this experiment, we uniformly set the model connection rate of different models to $0.1$. The model structure is $10$ times less than the original model parameters.
	
	As shown in Fig.~\ref{model_ADMM}, the training process of the original model has fluctuated. After the pruning is completed, the training process becomes stable. This shows that after subtracting part of the redundant structure, the training process of the model tends to be more stable than the original process. When the model connection rate is $0.1$, we find that the training process remains consistent with the unpruned situation. Even in some cases, the pruned model converges to a high accuracy after retraining. Fig.~\ref{model_ADMM} also shows the comparison results of MLP, GNN, CNN and SSGCNet (for Bonn and SSW) models, respectively. Our weight pruning method achieves $10$-time weight pruning on the deep learning model, and the accuracy remains basically unchanged. For the SSGCNet, there is almost no loss of accuracy. This experiment proves that the ADMM weight pruning strategy has universal applicability in different deep learning models, and it can maintain accuracy without reducing the number of model parameters.

	\subsubsection{Comparisons on other parameters}
	
	For comparison, we apply the WNFG strategy and the ADMM weight pruning method to the existing state-of-the-art methods. We test the performance of all the methods on graph complexity, number of parameters, and multiplication operation.
	
	In Table~\ref{SOTA_ADMM}, our proposed method has the high classification accuracy of the deep learning model. When the space occupancy rate of the graph representation is reduced by $10$ times, and the parameter amount of the 2DCNN model is reduced by $97$ times, the classification accuracy keeps stable. The minimum parameters of our model of ADMM weight pruning are shown in Table~\ref{Model_ADMM}. The experimental results show that the proposed model is suitable for a lightweight epileptic EEG signal classification tasks.

	\section{Conclusion}
	\label{sec:Conclusion}
	In this article, we have introduced a sparse spectra graph convolutional network (SSGCNet) method for EEG signal classification, which is based on ADMM-type splitting and weight pruning methods. We proposed an EEG signal graph representation method, a weighted neighborhood field graph (WNFG), which educed the data generation time and memory usage. We then have introduced an ADMM weight pruning method for compressing redundancy in SSGCNet. Our method has achieved a model connection rate of up to $97$ times in the Bonn dataset and the SSW dataset, respectively. Compared with other methods, our method has a lower computational cost and a smaller loss of classification accuracy.

	{\appendix[Proof of Lemma \ref{lemma:wb-sub}]
		\label{sec:proof}
		For the $\mathbf{w}_{1:N}$-subproblem, we get
		\begin{equation}  \label{eq:wb_sub}
			\begin{split}
				& \mathcal{L}(\mathbf{w}^{(k+1)}_{1:N},\mathbf{z}^{(k)}_{1:N};\bm{\eta}^{(k)}_{1:N})
				- 
				\mathcal{L}(\mathbf{w}^{(k)}_{1:N},\mathbf{z}^{(k)}_{1:N};\bm{\eta}^{(k)}_{1:N}) \\
				& =
				f (\mathbf{w}_{1:N}^{(k+1)}, \mathbf{y}) + g (\mathbf{z}_{1:N}^{(k)}) + (\bm{\eta}_{1:N}^{(k)})^{\T}(\mathbf{z}_{1:N}^{(k)}- \mathbf{\Omega}_{1:N}\mathbf{w}_{1:N}^{(k+1)}) \\
				&\qquad + \frac{\rho}{2} \left\| \mathbf{z}_{1:N}^{(k)}- \mathbf{\Omega}_{1:N}\mathbf{w}_{1:N}^{(k+1)} \right\|_2^2 \\
				& \quad - (
				f (\mathbf{w}_{1:N}^{(k)}, \mathbf{y}) + g ( \mathbf{z}_{1:N}^{(k)})  + (\bm{\eta}_{1:N}^{(k)})^{\T}(\mathbf{z}_{1:N}^{(k)}- \mathbf{\Omega}_{1:N}\mathbf{w}_{1:N}^{(k)}) \\
				&\qquad + \frac{\rho}{2} \left\| \mathbf{z}_{1:N}^{(k)}- \mathbf{\Omega}_{1:N}\mathbf{w}_{1:N}^{(k)} \right\|_2^2) \\
				&=
				f (\mathbf{w}_{1:N}^{(k+1)}, \mathbf{y}) - 
				f (\mathbf{w}_{1:N}^{(k)}, \mathbf{y}) \\
				&\quad + (\bm{\eta}_{1:N}^{(k)})^\T(\mathbf{\Omega}_{1:N}\mathbf{w}_{1:N}^{(k)}-\mathbf{\Omega}_{1:N}\mathbf{w}_{1:N}^{(k+1)})
				\\
				&\qquad + \frac{\rho}{2} \left\| \mathbf{z}_{1:N}^{(k)}- \mathbf{\Omega}_{1:N}\mathbf{w}_{1:N}^{(k+1)} \right\|_2^2  - 
				\frac{\rho}{2} \left\| \mathbf{z}_{1:N}^{(k)}- \mathbf{\Omega}_{1:N}\mathbf{w}_{1:N}^{(k)} \right\|_2^2 
				\\
				&{=} 
				f (\mathbf{w}_{1:N}^{(k+1)}, \mathbf{y}) - 
				f (\mathbf{w}_{1:N}^{(k)}, \mathbf{y}) 	
				+ \frac{\rho}{2} \left\| \mathbf{\Omega}_{1:N}\mathbf{w}_{1:N}^{(k)}- \mathbf{\Omega}_{1:N}\mathbf{w}_{1:N}^{(k+1)} \right\|_2^2  \\
				&\quad -  \langle \bm{\eta}_{1:N}^{(k)} + \rho  (\mathbf{z}_{1:N}^{(k)} - \mathbf{\Omega}_{1:N} \mathbf{w}_{1:N}^{(k)}),
				\mathbf{\Omega}_{1:N}\mathbf{w}_{1:N}^{(k)}- \mathbf{\Omega}_{1:N}\mathbf{w}_{1:N}^{(k+1)} \rangle
				)\\
				&=f (\mathbf{w}_{1:N}^{(k+1)}, \mathbf{y})\! -\! 
				f (\mathbf{w}_{1:N}^{(k)}, \mathbf{y}) 
				\!+\! \frac{\rho_1}{2} \left\| \mathbf{\Omega}_{1:N}\mathbf{w}_{1:N}^{(k)} \!-\! \mathbf{\Omega}_{1:N}\mathbf{w}_{1:N}^{(k+1)} \right\|_2^2 \\
				& \quad 	-  \langle  \partial_{\mathbf{w}}  f(\mathbf{w}_{1:N}^{(k+1)},\mathbf{y}),
				\mathbf{w}_{1:N}^{(k)}- \mathbf{w}_{1:N}^{(k+1)} \rangle.
			\end{split} 
		\end{equation}

		We assume that $\mathbf{\Omega}$ is full-column rank with $\mathbf{\Omega} \mathbf{\Omega}^\T \succeq \kappa^2 \mathbf{I}$, we have 
		\begin{equation}\label{eq:bound_dual_b}
			\begin{split}
				\begin{aligned}
					&\|  \mathbf{\Omega}_{1:N} \mathbf{w}_{1:N}^{(k+1)}  -  \mathbf{\Omega}_{1:N}\mathbf{w}_{1:N}^{(k)} \|^2  \\
					& =  \langle \mathbf{\Omega}_{1:N} (\mathbf{w}_{1:N}^{(k+1)} - \mathbf{w}_{1:N}^{(k)}) ,  \mathbf{\Omega}_{1:N} (\mathbf{w}_{1:N}^{(k+1)} - \mathbf{w}_{1:N}^{(k)}) \rangle \\
					& =  \langle  \mathbf{\Omega}_{1:N} \mathbf{\Omega}_{1:N}^\T (\mathbf{w}_{1:N}^{(k+1)} - \mathbf{w}_{1:N}^{(k)}) ,  (\mathbf{w}_{1:N}^{(k+1)} - \mathbf{w}_{1:N}^{(k)}) \rangle \\
					& \geq \kappa^2 \left\| \mathbf{w}_{1:N}^{(k+1)} - \mathbf{w}_{1:N}^{(k)} \right\|^2.
				\end{aligned}
			\end{split} 
		\end{equation}

		Combining \eqref{eq:prox-regular-case} and \eqref{eq:bound_dual_b}, we can obtain
		\begin{equation}  \label{eq:x_sub}
			\begin{split}
				& \mathcal{L}(\mathbf{w}^{(k+1)}_{1:N},\mathbf{z}^{(k+1)}_{1:N};\bm{\eta}^{(k+1)}_{1:N})
				- 
				\mathcal{L}(\mathbf{w}^{(k)}_{1:N},\mathbf{z}^{(k)}_{1:N};\bm{\eta}^{(k)}_{1:N}) \\
				& \geq 
				\frac{ M_{w} - \rho \kappa^2 }{2} 
				\left\| \mathbf{w}_{1:N}^{(k+1)}- \mathbf{w}_{1:N}^{(k)} \right\|^2 . 
			\end{split} 
		\end{equation}
	}

	\bibliographystyle{IEEEtran}
	\bibliography{bibfile}

\end{document}